\documentclass[pdflatex,sn-mathphys-ay]{sn-jnl}

\usepackage{graphicx}%
\usepackage{multirow}%
\usepackage{amsmath,amssymb,amsfonts}%
\usepackage{amsthm}%
\usepackage{mathrsfs}%
\usepackage[title]{appendix}%
\usepackage{xcolor}%
\usepackage{textcomp}%
\usepackage{manyfoot}%
\usepackage{booktabs}%
\usepackage{tabularx}
\usepackage{algorithm}%
\usepackage{algorithmicx}%
\usepackage{algpseudocode}%
\usepackage{listings}%

\theoremstyle{thmstyleone}%
\newtheorem{theorem}{Theorem}
\theoremstyle{thmstyletwo}%
\newtheorem{example}{Example}%
\newtheorem{lemma}{Lemma}%

\theoremstyle{thmstylethree}%

\begin{document}

\title[Article Title]{Limit of the Maximum Random Permutation Set Entropy}


\author[1,2]{\fnm{Jiefeng} \sur{Zhou}}

\author[3]{\fnm{Zhen} \sur{Li}}

\author[4]{\fnm{Kang Hao} \sur{Cheong}}

\author*[1,5]{\fnm{Yong} \sur{Deng}}\email{dengentropy@uestc.edu.cn, prof.deng@hotmail.com}

\affil*[1]{\orgdiv{Institute of Fundamental and Frontier Science}, \orgname{University of Electronic Science and Technology of China}, \orgaddress{\city{Chengdu}, \postcode{610054}, \state{Sichuan}, \country{China}}}

\affil[2]{\orgdiv{Yingcai Honors College}, \orgname{University of Electronic Science and Technology of China}, \orgaddress{\city{Chengdu}, \postcode{611731}, \state{Sichuan}, \country{China}}}

\affil[3]{\orgname{China Mobile Information Technology Center}, \orgaddress{\city{Beijing}, \postcode{100029}, \country{China}}}

\affil[4]{\orgdiv{Science, Mathematics and Technology Cluster}, \orgname{Singapore University of Technology and Design (SUTD)}, \orgaddress{\postcode{S487372}, \country{Singapore}}}

\affil[5]{\orgdiv{School of Medicine}, \orgname{Vanderbilt University}, \orgaddress{\city{Nashville}, \postcode{37240}, \state{Tennessee}, \country{USA}}}


\abstract{
The Random Permutation Set (RPS) is a new type of set proposed recently, which can be regarded as the generalization of evidence theory. To measure the uncertainty of RPS, the entropy of RPS and its corresponding maximum entropy have been proposed. Exploring the maximum entropy provides a possible way of understanding the physical meaning of RPS. In this paper, a new concept, the \textit{envelope} of entropy function, is defined. In addition, the limit of the \textit{envelope} of RPS entropy is derived and proved. Compared with the existing method, the computational complexity of the proposed method to calculate the \textit{envelope} of RPS entropy decreases greatly. The result shows that when $N \to \infty$, the limit form of the \textit{envelope} of the entropy of RPS converges to $e \times (N!)^2$, which is highly connected to the constant $e$ and factorial. Finally, numerical examples validate the efficiency and conciseness of the proposed \textit{envelope}, which provides a new insight into the maximum entropy function.
}




\maketitle

Uncertainty management is a significant issue that has attracted a lot of interest in various kinds of research fields. The classical tool to deal with uncertainty is probability theory \citep{jaynes2003probability}, which allocates the probability distribution defined in a mutually exclusive event space. However, one particular challenge arises when uncertain information needs to be defined on the power set of the event space, which cannot be effectively handled by probability theory. To address this issue, Dempster-Shafer evidence theory (DSET) \citep{dempster2008upper, shafer1976mathematical} is developed. DSET generalized probability theory by extending the probability distribution to a mass function defined on the collection of all possible subsets. As a result, this extension enabled the application of DSET to domains characterized by a high degree of complexity and uncertainty \citep{Xiao2022GQET, Xiao2022Acomplexweighted, Xiao2023QuantumXentropy, Xiao2023NQMF}. But both theories overlook the significance of considering ordered information in processing uncertain information \citep{he2021mmget}. This aspect, however, holds substantial importance and should not be ignored. Thus, Deng \citep{deng2022random} proposed the random permutation set (RPS), considering the ordered information while fully compatible with DSET and probability theory.

To measure the uncertainty, Shannon entropy \citep{shannon1948mathematical} is used in probability theory. In DSET, Deng \citep{deng2016deng} proposed Deng entropy. In RPS, Chen and Deng \citep{chen2023entropy} proposed RPS entropy. Each entropy offers an efficient approach to understanding uncertainty. The maximum entropy principle was extensively studied by Jaynes \citep{jaynes1957information,jaynes1982rationale}. The principle asserts that the distribution with the maximum entropy is the most appropriate representation of a system's current state. For convenience, this paper defines the concept of \textit{envelope} for entropy. It refers to the function inside the logarithmic function in the overall entropy expression for the maximum entropy case of a system. This concept becomes particularly significant when considering specific examples or scenarios. For instance, in a thermodynamic system, the \textit{envelope} of entropy can represent the range of possible energy states that lead to the highest level of disorder or randomness. Understanding the \textit{envelope} of entropy provides valuable insights into the behavior and characteristics of the system under consideration. This principle has been applied to various kinds of practical applications and theories, such as emergency management \citep{che2022maximum}, power law distribution \citep{yu2022derive}, and negation transformation \citep{deng2020negation}. As shown in Fig.~\ref{fig.envelope}, for a sample space whose cardinality is $N$, the \textit{envelope} of Shannon entropy of a probability distribution is $N$, while in a power set whose maximum cardinality of its subsets is $N$, the \textit{envelope} of Deng entropy is $3^N-2^N$. Although the analytical expression of the \textit{envelope} of RPS entropy is given by Deng and Deng \citep{deng2022maximum}, its computational and expression complexity makes it difficult and inconvenient to discuss the properties and its physical meaning of maximum RPS entropy.

\begin{figure}
\centering
\includegraphics[width=0.9\textwidth]{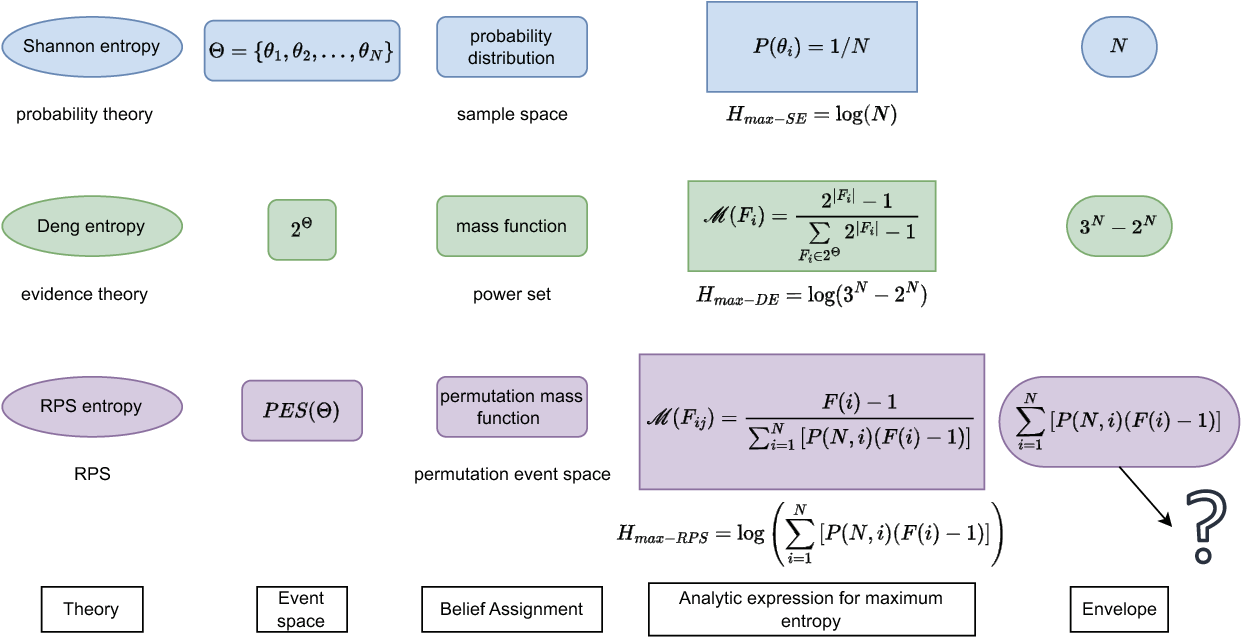}
\caption{The connection between Shannon entropy \citep{shannon1948mathematical}, Deng entropy \citep{deng2016deng}, and RPS entropy \citep{chen2023entropy}. The \textit{envelope} is the function in the logarithmic function in maximum entropy expression.}\label{fig.envelope}
\end{figure}

To address this issue, the limit form of the \textit{envelope} of RPS entropy is presented and proved in this paper. The limit $e \times (N!)^2$ is very concise and related to the natural constant $e$ and factorial, whose physical meaning may be an interesting topic in future research. Besides, numerical analysis validates the conciseness and correctness of the result.

The remainder of this paper is structured as follows. First, Sec.~\ref{se.preliminaries} introduces preliminary information and definitions related to the work. Next, Sec.~\ref{se.envelope} presents the definition and an illustrative example of the concept of \textit{envelope}. After that, Sec.~\ref{se.proof} provides and proves the limit form of the entropy \textit{envelope} for the RPS entropy. To supplement this theoretical analysis, Sec.~\ref{se.examples} then gives a comparative analysis between three types of maximum entropies. Finally, Sec.~\ref{se.colclusion} summarizes the key conclusions of the paper.

\section{\label{se.preliminaries}Preliminaries}

Some preliminaries are introduced in this section.

\subsection{Mass function in power set}
To better model and reason the uncertainty in the real world, many theories and approaches have been proposed to address this issue, among which Dempster-Shafer evidence theory (DSET) \citep{dempster2008upper, shafer1976mathematical} is an effective method for uncertainty reasoning and information fusion \citep{Xiao2022Generalizeddivergence,zeng2023new}.

\subsubsection{\label{def:power_set}Power set}
Let $\Omega$ be the frame of discernment (FOD), expressed as $\Omega = \{\chi_1, \chi_2, \ldots, \chi_N\}$, whose elements are mutually exclusive and exhaustive. Its corresponding \textbf{power set} $2^\Omega$ contains all the subsets of $\Omega$, denoted as

    \begin{equation}
        2^\Omega = \{ \emptyset, \{\chi_1\}, \{\chi_2\}, \ldots, \{\chi_N\}, \{\chi_1,\chi_2 \}, \{\chi_1,\chi_3 \},\ldots, \Omega \}.
    \end{equation}

Since the cardinality of $2^\Omega$ is $2^N$, which is highly related to Sierpinski gasket, a recent study proposed by Zhou and Deng \citep{zhou2023generating} also reveals this connection.

\subsubsection{\label{def:mass_function}Mass function}
The \textbf{mass function} or basic probability assignment (BPA), is a mapping $\mathscr{M}: 2^\Omega \to [0,1]$ \citep{dempster2008upper, shafer1976mathematical} with bound conditions:
    \begin{equation}
\mathscr{M}(\emptyset) = 0, \quad \sum\limits_{M_i \in 2^\Omega } \mathscr{M}(M_i)=1.
    \end{equation}

\subsection{Deng entropy}

To measure uncertainty of a given system, several methods and entropies have been proposed, including $\chi$ entropy \citep{Xiao2021maximum} for the complex-valued distribution, belief structure method \citep{cui2022belief} for physiological signals, betweenness structural entropy \citep{zhangqi2022CSF} and local structure entropy \citep{lei2022node} for complex networks, an improved belief entropy \citep{chen2023newentropy} and divergence method \citep{gao2022improved, lyu2024belief} for decision making, generalized weighted permutation entropy for biomes analysis \citep{stosic2022generalized,stosic2024generalized}. Some methods, such as the Deng entropy \citep{deng2016deng}, though have been subject to critique \citep{Abell2017Analyzing, moral2020critique}, continue to be widely used in many areas and theories, including multi-criteria decision-making \citep{Xiao2020EFMCDM}, information fusion \citep{Xiao2021GIQ, Xiao2022GEJS}, pattern recognition \citep{kazemi2021fractional}, decisions prediction and evaluation \citep{lai2022comprehensive}, fractal dimension analysis \citep{qiang2021information}, information dimension analysis \citep{lei2022information}, and prior distribution \citep{li2023normal}.

\subsubsection{\label{def:deng_entropy}Deng entropy}
Given a mass function $\mathscr{M}(2^\Omega)$, \textbf{Deng entropy} \citep{deng2016deng} is defined as
    \begin{equation}
H_{DE}(\mathscr{M}) = - \sum\limits_{M_i \in 2^ \Omega} \mathscr{M}(M_i) \log \frac{\mathscr{M}(M_i)}{2^{|M_i|}-1},
    \end{equation}
where $|M_i|$ is the cardinality of $M_i$. When the element in each mass function is a singleton set, i.e. $\forall M_i \in 2^\Omega, |M_i|=1$, then Deng entropy degenerates into \textbf{Shannon entropy} \citep{shannon1948mathematical}:
\begin{equation}\label{eq.Shannon}
    H_{SE}(P) = - \sum\limits_{p_i \in P} p_i \log (p_i).
\end{equation}

\begin{theorem}[Maximum Deng entropy\citep{kang2019maximum}]
Given a FOD $\Omega$, the \textbf{maximum Deng entropy} $H_{max-DE}$  can be obtained when its mass function has the following form:
    \begin{equation}
        \mathscr{M}(M_i)=\frac{2^{|M_i|}-1}{\sum\limits_{M_i \in 2^\Omega}2^{|M_i|}-1}.
    \end{equation}
    Its corresponding Deng entropy reaches its maximum:

    \begin{align}
        H_{max-DE} &= \sum\limits_{M_i \in 2^\Omega} \mathscr{M}(M_i)\log \frac{\mathscr{M}(M_i)}{2^{|M_i|}-1} \nonumber \\
        &= \log \sum\limits_{M_i \in 2^\Omega} (2^{|M_i|}-1).
    \end{align}
\end{theorem}

\subsection{Random Permutation Set}

\textbf{Random Permutation Set (RPS)} has been presented as a means to manage uncertainty with ordered information \citep{deng2022random}. RPS entails permuting items in a given set, allowing for effective handling of uncertainty in ordered data. Based on the above, some works like the distance of RPS \citep{chen2023distance} and the fusion order \citep{chen2023permutation} are developed. Some fundamental definitions of RPS are introduced briefly here.

\subsubsection{\label{def:PES}Permutation Event Space (PES)}
Given a finite set with $N$ elements $\Omega = \{\chi_1, \chi_2,\ldots, \chi_N \}$, its \textbf{Permutation Event Space (PES)} \citep{deng2022random} is a \textbf{ordered} set containing all possible permutations of all subsets of $\Omega$.

\begin{align}
    PES(\Omega)= & \left\{M_{i,j} \mid i=0, \ldots, N ; j=1, \ldots, A_{N}^{i}\right\}                                                                    \nonumber \\
    =            & \left\{\emptyset,\left(\chi_1\right),\left(\chi_2\right), \ldots,\left(\chi_N\right),\left(\chi_1, \chi_2\right),\right. \nonumber\\
    & \left(\chi_2, \chi_1\right), \ldots,\left(\chi_{N-1}, \chi_N\right),  \left(\chi_N, \chi_{N-1}\right), \nonumber\\
    & \left.\ldots,\left(\chi_1, \chi_2, \ldots, \chi_N\right), \ldots,\left(\chi_N, \chi_{N-1}, \ldots, \chi_1\right)\right\}.
\end{align}

Where $A_{N}^{i}={N!}/{(N-i)!}$ is the number of choices to select $i$ ordered elements from a collection with $N$ elements. The element $M_{i,j}$ is called \textbf{permutation event}.

\subsubsection{\label{def:RPS}Random Permutation Set (RPS)}
Given a finite set with $N$ elements $\Omega = \{\chi_1, \chi_2,\ldots, \chi_N\}$, its \textbf{Random Permutation Set (RPS)} \citep{deng2022random} is a set of pairs.

\begin{equation}
    RPS(\Omega)=\{\langle M_{i,j}, \mathcal{M}(M_{i,j})\rangle \mid M_{i,j} \in P E S(\Omega)\}.
\end{equation}

Where $\mathcal{M}$ is called \textbf{permutation mass function (PMF)}, a mapping: $PES(\Omega) \to [0,1]$ with bound conditions:

\begin{gather}
    \mathcal{M}(\emptyset) = 0, \quad \sum\limits_{M_{i,j} \in PES(\Omega) } \mathcal{M}(M_{i,j})=1.
\end{gather}

RPS is completely consistent with DSET and probability theory. When the order of elements within PES is disregarded, PES becomes a power set, PES effectively becomes a power set, while the PMF within RPS reduces to the mass function. Moreover, if each permutation event contains only a singular element, RPS simplifies into probability theory, wherein PES degenerates into the sample space. Fig.~\ref{fig.rel_RPS} illustrates the connection between RRS, DSET, and probability theory.

\begin{figure}[h]
\centering
\includegraphics[width=0.9\textwidth]{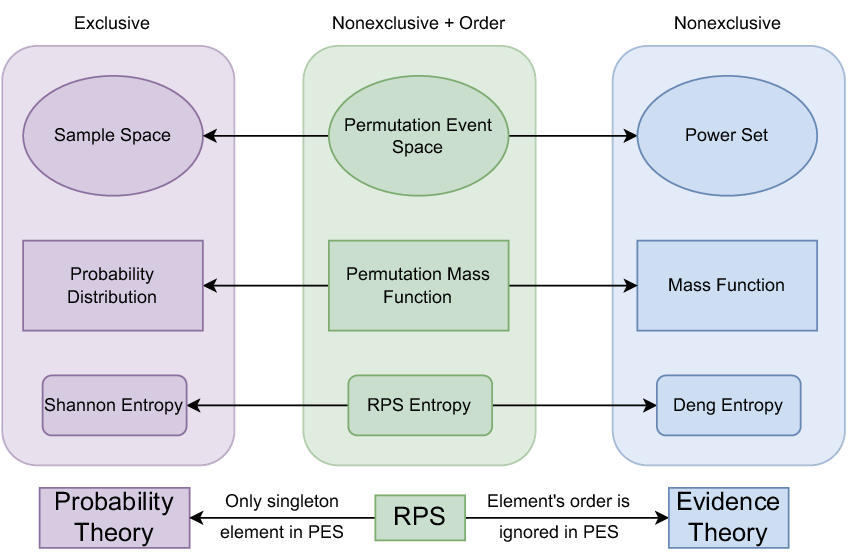}
\caption{The relationship between RPS, DSET and probability theory.}\label{fig.rel_RPS}
\end{figure}

To ascertain the degree of uncertainty in RPS, the entropy of RPS, as introduced by Chen and Deng \citep{chen2023entropy}, plays a significant role, exhibiting compatibility with both Deng entropy and Shannon entropy.

\subsubsection{\label{def:RPS_entropy}RPS entropy}
For a RPS: $RPS(\Omega)=\{\langle M_{i,j}, \mathcal{M}(M_{i,j})\rangle \mid M_{i,j} \in P E S(\Omega)\}$, defined on a PES: $PES(\Omega)=\left\{M_{i,j} \mid i=0, \ldots, N ; j=1, \ldots, A_{N}^{i}\right\}$, the entropy of RPS \citep{chen2023entropy} is defined as follows.

\begin{equation}
H_{RPS}(\mathcal{M})=-\sum_{i=1}^N \sum_{j=1}^{A_{N}^{i}} \mathcal{M}\left(M_{i,j}\right) \log \left(\frac{\mathcal{M}\left(M_{i,j}\right)}{S_{A}(i)-1}\right).
\end{equation}
Where $S_{A}(i)=\sum\limits_{a=0}^i A_{i}^{a} = \sum\limits_{a=0}^i \frac{i!}{(i-a)!}$ is the sum of all possible permutations of a finite set containing $i$ elements.

Though the entropy of RPS is introduced, there is a lack of in-depth analysis regarding its maximum entropy. Therefore, Deng and Deng \citep{deng2022maximum} presented an analytical expression for the maximum entropy of RPS as well as its PMF condition.

\begin{theorem}[Maximum RPS entropy \citep{deng2022maximum}]\label{th.maximumRPS}
Given a PES: $PES(\Omega)=\left\{M_{i,j} \mid i=0, \ldots, N ; j=1, \ldots, A_{N}^{i}\right\}$, if and only if its PMF satisfies

    \begin{equation}
\mathcal{M}(M_{i,j})=\frac{S_{A}(i)-1}{\sum_{i=1}^N \left[A_{N}^{i}(S_{A}(i)-1)\right]}.
    \end{equation}
    Then the entropy of RPS reaches its maximum:

    \begin{equation}\label{eq.max_RPS}
H_{max-RPS}=\log\left( \sum\limits_{i=1}^{N} \left[ A_{N}^{i} (S_{A}(i)-1)\right]\right).
    \end{equation}
\end{theorem}

\section{\label{se.envelope}The envelope of entropy}

In this section, the definition of the \textit{envelope} of entropy is given, followed by an example as an illustration.

Within a system, the value of entropy can vary across multiple situations. In this context, the \textit{envelope} of entropy refers to the function encapsulated within the logarithmic expression of entropy that reaches its maximum value.

\subsection{\label{def:envelope}The definition of envelope of entropy}
    For a given definition of entropy $H_{entropy}=\mathbb{E}(-\log(f(P)))$, where $P$ is a belief assignment within a system $\mathcal{S}$ and $\mathbb{E}(X)$ is the mathematical expectation of $X$. Then the \textit{envelope} of $H_{entropy}$ can be defined as follows.

    \begin{align}
        H_{max-entropy} = \max_{P \in \mathcal{S}} \left[\mathbb{E}(-\log(f(P)))\right], \\
        C_e(H_{entropy}) = \exp (H_{max-entropy}) \label{eq.envelope}.
    \end{align}
    Where $H_{max-entropy}$ is the maximum entropy for a given system, and $C_e$ is the \textit{envelope} of a given entropy $H_{entropy}$.

\begin{example}[The envelope of Shannon entropy and Deng entropy]
    Given a finite set $\Omega = \{\chi_1, \chi_2, \ldots, \chi_N\}$, calculate
    \begin{enumerate}
        \item the \textit{envelope} of Shannon entropy \citep{shannon1948mathematical};
        \item the \textit{envelope} of Deng entropy \citep{deng2016deng} in power set $2^\Omega$.
    \end{enumerate}
\end{example}

For a finite set $\Omega$, Shannon entropy reaches its maximum when the probability distribution is uniform, i.e. $P(\chi_i)=1/N$. Then Eq.(\ref{eq.Shannon}) can be rewritten as:

\begin{equation}
    H_{SE}(P)=\log (N).
\end{equation}

Based on Eq.(\ref{eq.envelope}), the \textit{envelope} of Shannon entropy is

\begin{equation}
    C_e(H_{S}) = \exp \left(\log (N)\right) = N.
\end{equation}

If the mass function satisfies \citep{kang2019maximum}

\begin{equation}
    \mathscr{M}(M_i)=\frac{2^{|M_i|}-1}{\sum_{M_i \in 2^\Omega}2^{|M_i|}-1},
\end{equation}

then Deng entropy reaches its maximum, which can be simplified as \citep{qiang2023multifractal}

\begin{align}
    H_{max-DE} &= \sum\limits_{M_i \in 2^\Omega} \mathscr{M}(M_i)\log \frac{\mathscr{M}(M_i)}{2^{|M_i|}-1} \nonumber \\
    &= \log \sum\limits_{M_i \in 2^\Omega} (2^{|M_i|}-1) \nonumber \\
    &= \log \left( \sum\limits_{a=0}^{N}\left(C_{N}^{a}(2^a-1)\right)  \right) \nonumber \\
    & = \log \left(3^N - 2^N  \right),
\end{align}
where $C_{N}^{a}={N!}/{[(N-a)! \times a!]}$ is the combination number.

Based on Eq.(\ref{eq.envelope}), the \textit{envelope} of Deng entropy is

\begin{equation}
    C_e(H_{DE}) = \exp \left(\log \left(3^N - 2^N  \right)\right) = 3^N - 2^N.
\end{equation}

As shown in Fig.~\ref{fig.def_envelope}, given a uniform probability distribution, $H_{SE}$ reaches its maximum $\log (N)$, then the \textit{envelope} of $H_{SE}$ is $N$. When the mass function satisfies $\mathscr{M}(M_i)={2^{|M_i|}-1}/{[\sum_{M_i \in 2^\Omega}2^{|M_i|}-1]}$, $H_{DE}$ reaches its maximum: $\log \left(3^{N} - 2^{N}\right)$, then the \textit{envelope} of $H_{DE}$ is $\left(3^{N} - 2^{N}\right)$.


\begin{figure}
\centering
\includegraphics[width=0.9\textwidth]{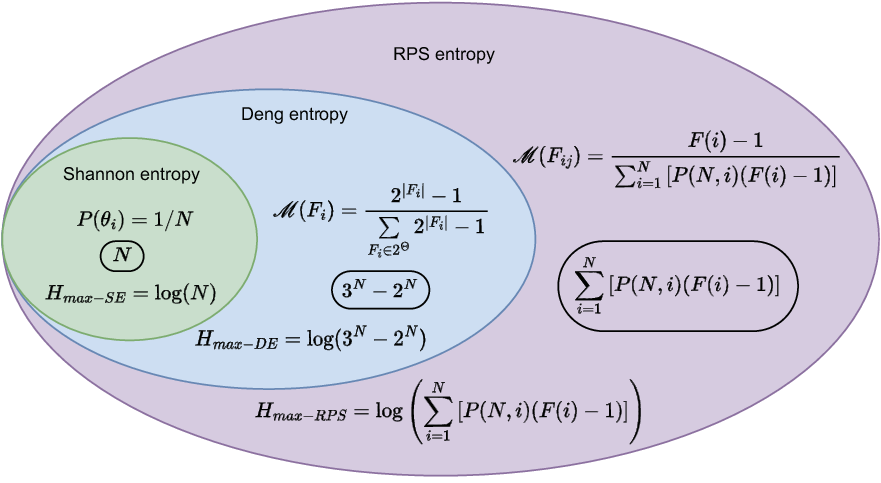}
\caption{\label{fig.def_envelope}The \textit{envelope} of Shannon entropy, Deng entropy and RPS entropy.}
\end{figure}

\begin{theorem}[Envelope of entropy]\label{th.envelope}
The \textit{envelope} of Shannon entropy, Deng entropy, and RPS entropy are presented as follows.
\begin{gather}
    \text{Shannon entropy}:C_e(H_{SE})=N,\\
    \text{Deng entropy}: C_e(H_{DE})=3^N-2^N,\\
    \text{RPS entropy}: C_e(H_{RPS}) = \sum\limits_{a=1}^{N} \left[ A_{N}^{a} (S_{A}(a)-1)\right].
\end{gather}
\end{theorem}

\section{Limit of the envelope of RPS entropy}\label{se.proof}

The proof of the limit of the \textit{envelope} of RPS entropy is presented in this section. Let $A_{N}^{u}$ marked as the $u$-permutation of $N$, and $S_{A}(N)$ marked as the sum of $A_{N}^{u}$ when $u$ change from $0$ to $N$.

\begin{align}
    A_{N}^{u} &=\frac{N!}{(N-u)!},                                            \\
S_{A}(N) & = \sum\limits_{u=0}^N A_{N}^{u} = \sum\limits_{u=0}^N \frac{N!}{(N-u)!}.
\end{align}

\begin{lemma}[sum of permutation]\label{le.sum_of_permutation}
When $N \geqslant 1$, $S_{u}(N)$ can be rewritten as

    \begin{equation}
S_{A}(N) = \sum\limits_{u=0}^N A_{N}^{u} = \lfloor e \times N! \rfloor ~ (N \geqslant 1),
    \end{equation}

    where $e$ is the nature constant and $\Gamma(N+1,1)$ is the "upper" incomplete gamma function:
    \begin{equation}
        \Gamma(u+1,x)=\int_x^\infty t^{u} e^{-t} \mathrm{d} t.
    \end{equation}

\end{lemma}

\begin{proof}[Proof of Lemma \ref{le.sum_of_permutation}]

For $S_{A}(N)$, it can be rewritten as
\begin{equation}\label{eq.S_{A}(N)}
S_{A}(N)=\sum\limits_{u=0}^N A_{N}^{u}=N!\left(\sum\limits_{u=0}^N\frac{1}{u!}\right).
    \end{equation}

    For a non-negative integer $N$, the `upper' incomplete gamma function can be rewritten as \citep{arfken2011mathematical}:
    \begin{equation}\label{eq.Gamma}
        \Gamma(N+1,x) =\int_x^\infty t^{N} e^{-t}\mathrm{d} t = N! e^{-x} \sum\limits_{u=0}^{N} \frac{x^u}{u!}.
    \end{equation}

Let $x=1$ in Eq.(\ref{eq.Gamma}):
    \begin{equation}\label{eq.Gamma_x=1}
        \Gamma(N+1,x)|_{x=1} = \frac{N!}{e} \sum\limits_{u=0}^{N} \frac{1}{u!}.
    \end{equation}

Combined with Eq.(\ref{eq.Gamma_x=1}) and Eq.(\ref{eq.S_{A}(N)}), $S_{A}(N)$ can be simplified as:
\begin{equation}\label{eq.reS_{A}(N)}
S_{A}(N) = \sum\limits_{u=0}^N A_{N}^{u} = N!\left(\sum\limits_{u=0}^N\frac{1}{u!}\right)= e \times \Gamma(N+1,1).
    \end{equation}

Since $\forall N \in \mathbb{Z^+}, u=0,1,\ldots, N$,   $A_{N}^{u}$ is an integer, $S_{A}(N)$ can be simplified as:

\begin{equation}\label{eq.SimS_{A}(N)}
S_{A}(N) = e \times \Gamma(N+1,1) = \lfloor e \times N! \rfloor,
    \end{equation}

    where $\Gamma(N+1,1) = \frac{\lfloor e \times N! \rfloor}{e}$ \citep{arfken2011mathematical}.

\end{proof}

For convenience, let's define a function:

\begin{equation}\label{eq.S(N)}
S(N)=\sum\limits_{u=1}^{N} [ A_{N}^{u} (S_{A}(u)-1)].
\end{equation}

Then according to \autoref{th.maximumRPS}, the maximum entropy of RPS can be rewritten as:

\begin{equation}\label{eq.h_max-RPS}
    H_{max-RPS}(N)=\log S(N).
\end{equation}

\begin{lemma}[Approximation of $S(N)$]\label{le.appropriation}
\end{lemma}

\begin{align}
        e \times N! &\left(\sum_{u=1}^{N} \frac{u!}{(N-u)!} -2\right) +1 \leqslant \lim_{N \to \infty} S(N) \nonumber\\
        &\leqslant e \times N! \left(\sum_{u=1}^{N} \frac{u!}{(N-u)!} -1\right) +2 \label{eq.S(q)inequ}.
\end{align}

\begin{proof}[Proof of Lemma \ref{le.appropriation}]
    Based on Lemma \ref{le.sum_of_permutation} and Eq.(\ref{eq.S(N)}), $S(N)$ can be rewritten as:

    \begin{align}
S(N) & = \sum\limits_{u=1}^{N} [A_{N}^{u} (S_{A}(u)-1)] \nonumber                                                                     \\
& = \sum\limits_{u=1}^{N}[A_{N}^{u} (S_{A}(u))] - \sum\limits_{u=1}^{N}A_{N}^{u} \nonumber                                         \\
& = \sum\limits_{u=1}^{N}[A_{N}^{u} (S_{A}(u))] - \sum\limits_{u=0}^{N}A_{N}^{u} +1 \nonumber                                      \\
             & = \sum\limits_{u=1}^{N}[ \frac{N!}{(N-u)!} \lfloor e \times u! \rfloor] - \lfloor e \times N! \rfloor +1 \nonumber            \\
             & = N! \sum\limits_{u=1}^{N} \frac{\lfloor e \times u! \rfloor}{(N-u)!}  - \lfloor e \times N! \rfloor +1 \label{eq.S(N)_align}.
    \end{align}

    Suppose that
    \begin{equation}
        e \times u! = \lfloor e \times u! \rfloor + \varepsilon_1, \quad e \times N! = \lfloor e \times N! \rfloor + \varepsilon_2,
    \end{equation}

    where $ \varepsilon_1, \varepsilon_2 \in [0,1)$.

    Then Eq.(\ref{eq.S(N)_align}) can be rewritten as:

    \begin{align}
        S(N) & = N! \sum\limits_{u=1}^{N} \frac{\lfloor e \times u! \rfloor}{(N-u)!}  - \lfloor e \times N! \rfloor +1 \nonumber                                                                   \\
             & = N! \sum\limits_{u=1}^{N} \frac{e \times u! - \varepsilon_1}{(N-u)!} - (e \times N! - \varepsilon_2) +1 \nonumber                                                                  \\
             & = e \times N! \left(\sum\limits_{u=1}^{N} \frac{u!}{(N-u)!} -1\right)  -\varepsilon_1 \times N! \sum\limits_{j=0}^{N-1} \frac{1}{j!}+\varepsilon_2+ 1 \label{eq.S(N)_align_epsilon}.
    \end{align}

Based on Eq.(\ref{eq.S_{A}(N)}), the $\varepsilon_1 \times N! \sum\limits_{j=0}^{N-1} \frac{1}{j!}$ in Eq.(\ref{eq.S(N)_align_epsilon}) can be simplified as:

    \begin{align}
        \varepsilon_1 \times N! \sum\limits_{j=0}^{N-1} \frac{1}{j!} & = \varepsilon_1 \times N \times (N-1)! \sum\limits_{j=0}^{N-1} \frac{1}{j!} \nonumber                     \\
& = \varepsilon_1 \times N \times e \times \Gamma (N, 1) \nonumber                                           \\
& =  \varepsilon_1 \times e \times N! -\varepsilon_1 \varepsilon_3 \times N              \label{eq.enapprox},
    \end{align}
    where $\varepsilon_3 \in [0,1)$.

    Combined Eq.(\ref{eq.S(N)_align_epsilon}) and Eq.(\ref{eq.enapprox}), $S(N)$ can be rewritten as

    \begin{align}
        S(N)= &e \times N! \left(\sum\limits_{u=1}^{N} \frac{u!}{(N-u)!} -1\right) \nonumber\\
        &-\varepsilon_1 \times e \times N! + \varepsilon_1 \varepsilon_3 \times N +\varepsilon_2+ 1.
    \end{align}

    Since $\varepsilon_1, \varepsilon_2, \varepsilon_3 \in [0,1)$,
    then let $\varepsilon_1 \to 1, \varepsilon_2=\varepsilon_3=0$, $S(N)$ is no less than the following equation:
    \begin{equation}
        S(N) \geqslant e \times N! \left(\sum\limits_{u=1}^{N} \frac{u!}{(N-u)!} -2\right) +1 \label{eq.S(N)_geq}.
    \end{equation}
    Similarly, let $\varepsilon_1=0, \quad \varepsilon_2,\varepsilon_3 \to 1$, $S(N)$ is no greater than the following equation:

    \begin{align}\label{eq.S(N)_leq}
        S(N) \leqslant & e \times N! \left(\sum\limits_{u=1}^{N} \frac{u!}{(N-u)!} -1\right) + 2.
    \end{align}

    Based on Eq.(\ref{eq.S(N)_geq}) and Eq.(\ref{eq.S(N)_leq}), Lemma \ref{le.appropriation} is proved.

\end{proof}

\begin{lemma}[Limit of $\sum_{u=1}^{N} \frac{u!}{(N-u)!}$]\label{le.limit}
    \begin{equation}
        \lim\limits_{N \to\infty} \frac{\sum\limits_{u=1}^{N} \frac{u!}{(N-u)!}}{N!}-1 = 0\label{eq.limofsum}.
    \end{equation}

\end{lemma}

\begin{proof}[Proof of Lemma \ref{le.limit}]

Let $N=2k$, Then $\frac{u!}{(N-u)!}=\frac{1}{1}=1$. Thus, $\sum_{u=1}^{N} \frac{u!}{(N-u)!} \geqslant 1$.

    \begin{align}
        0 & \leqslant \lim_{N \to +\infty} \frac{\sum\limits_{u=1}^{N}\frac{u!}{(N-u)!}}{N!}-1\nonumber                                                 \\
          & =  \lim_{N \to +\infty} \sum\limits_{u=1}^{N} \frac{u!}{	N!(N-u)!}-1\nonumber                                                                \\
          & = \lim_{N \to +\infty} \left(\sum\limits_{u=1}^{N-3} \frac{u!}{	N!(N-u)!}\right) + \left.\frac{u!}{	N!(N-u)!}\right|_{u=N-2,N-1,N}-1\nonumber \\
          & = \lim_{N \to +\infty} \sum\limits_{u=1}^{N-3}\frac{u!}{N!(N-u)!}+ \frac{1}{	2N(N-1)}+\frac{1}{N}\label{eq.lim}.
    \end{align}

    Noted that $\frac{u!}{N!(N-u)!}$ is a monotonic series increasing as $u$ increases. So $\sum_{u=1}^{N-3}\frac{u!}{N!(N-u)!}$ is no greater than

    \begin{equation}\label{eq.ineq1}
        \sum\limits_{u=1}^{N-3}\frac{u!}{N!(N-u)!} \leqslant (N-3)\left.\frac{u!}{	N!(N-u)!}\right|_{u=N-2}=\frac{N-3}{2N(N-1)}.
    \end{equation}

    Based on Eq.(\ref{eq.ineq1}), Eq.(\ref{eq.lim}) can be simplified as:

    \begin{align}
        \lim_{N \to +\infty} & \sum_{u=1}^{N-3} \frac{u!}{N!(N-u)!}+ \frac{1}{2N(N-1)}+\frac{1}{N} & \nonumber \\ & \leqslant \lim_{N \to +\infty} \frac{N-2}{2N(N-1)} + \frac{1}{N} \nonumber \\
        & \quad = 0
    \end{align}

    Therefore, Eq.(\ref{eq.lim}) can be simplified as:
    \begin{align}
        0  \leqslant \lim_{N \to +\infty} \frac{\sum\limits_{u=1}^{N}\frac{u!}{(N-u)!}}{N!}-1\nonumber \leqslant 0.
    \end{align}

    Thus, the limit form of $\sum_{u=1}^{N} \frac{u!}{(N-u)!}$ is 
    \begin{equation}
        \lim\limits_{N \to \infty} \sum_{u=1}^{N} \frac{u!}{(N-u)!}=N!.
    \end{equation}

    Lemma \ref{le.limit} is proved.

\end{proof}

\begin{theorem}[the limit form of the envelope of RPS entropy]\label{th.limS(N)}
    \begin{align}
        \lim\limits_{N\to \infty} \mathop{C_e(H_{RPS}(\mathcal{M}))}\limits_{\mathcal{M}:PES(\Omega) \to [0,1], |\Omega|=N} &=\lim\limits_{N\to \infty} \sum\limits_{u=1}^{N} \left[ A_{N}^{u} (S_{A}(u)-1)\right] \nonumber\\
        & = \lim\limits_{N\to \infty} S(N) =  e \times (N!)^2.
    \end{align}
\end{theorem}
\begin{proof}[Proof of Lemma \ref{th.limS(N)}]
    According to Lemma \ref{le.appropriation} and Lemma \ref{le.limit}, When taking the limit at both ends of Eq.(\ref{eq.S(q)inequ}):
    \begin{align}
        \lim_{N \to \infty} e \times N! &\left(\sum_{u=1}^{N} \frac{u!}{(N-u)!} -2\right) +1 \leqslant \lim_{N \to \infty} S(N) \nonumber\\
        &\leqslant \lim_{N \to \infty} e \times N! \left(\sum_{u=1}^{N} \frac{u!}{(N-u)!} -1\right) +2 \label{eq.limends},
        \end{align}
    the limit form of both ends in Eq.(\ref{eq.limends}) is $e \times (N!)^2$, namely,

    \begin{align}\label{eq.limitofS}
        e \times (N!)^2 \leqslant \lim\limits_{N \to \infty}S(N) \leqslant e \times (N!)^2 \nonumber \\
        \Longrightarrow \lim\limits_{N \to \infty}S(N)=e \times (N!)^2.
    \end{align}

    Based on \autoref{th.maximumRPS}, Eq.(\ref{eq.envelope}), Eq.(\ref{eq.S(N)}), Eq.(\ref{eq.h_max-RPS}) and Eq.(\ref{eq.limitofS}), the limit form of the \textit{envelope} of RPS entropy is proved as follows.

    \begin{align}
        \lim\limits_{N\to \infty} \mathop{C_e(H_{RPS}(\mathcal{M}))}\limits_{\mathcal{M}:PES(\Omega) \to [0,1], |\Omega|=N} &= \lim\limits_{N\to \infty} \exp \left( H_{max-RPS} \right) \nonumber\\
        &= \lim\limits_{N\to \infty}  \sum\limits_{u=1}^{N} \left[ A_{N}^{u} (S_{A}(u)-1)\right] \nonumber \\
        &= \lim\limits_{N\to \infty} S(N) \nonumber\\
        & = e \times (N!)^2.
    \end{align}

    \autoref{th.limS(N)} is proved.
\end{proof}

\section{Numerical examples and discussion}\label{se.examples}

This section gives some examples to demonstrate the limit discussed earlier. Furthermore, it delves into the relationship between Shannon entropy \citep{shannon1948mathematical}, Deng entropy \citep{deng2016deng}, and RPS entropy \citep{chen2023entropy}, specifically examining their maximum values.

\begin{example}\label{ex.stirling}
In this example, the comparison between the maximum RPS entropy and the presented limit is on a large scale. Apart from that, the errors between factorial and its approximation--Stirling's formula \citep{tweddle2003james},
    \begin{equation}
        N! \approx \sqrt{2 \pi N}\left(\frac{N}{e}\right)^N \overset{def}{\longrightarrow} S_t(N),
    \end{equation}
are also presented here to better illustrate the presented limit.
\end{example}

Similar to Stirling's formula which is regarded as an approximation of factorial, the approximation of the \textit{envelope} of RPS entropy is denoted as $S_lim(N)=e\times (N!)^2$. The values of $S(N)$, $H_{max-RPS}$, and their corresponding proposed estimation $S_{lim}(N)$, $H_{lim-RPS}=\log S_{lim}(N)$, are listed in Tab.~\ref{tab.RPSscale}. To better illustrate the efficiency and conciseness of the limit form of the \textit{envelope} of RPS entropy, the results of factorial, Stirling's formula, and its absolute errors, relative errors are also listed in Tab.~\ref{tab.Stir}.

\begin{table}
    \caption{\label{tab.RPSscale}Value of $S(N), H_{max-RPS}$, and their corresponding estimation $S_{lim}(N), H_{lim-RPS}$ with different values of $N$. The error of $S_{lim}(N), H_{lim-RPS}$ are marked as $\Delta S, \Delta H$, respectively. While the subscripts in them indicate the relative error or absolute error. }
    \begin{tabular}{lccccccccc}
        $N$ & $S(N)$    & $S_{lim}(N)$ & $\Delta S_{abs}$ & $\Delta S_{rel}$ & $H_{max-RPS}$ & $H_{lim-RPS}$ & $\Delta H_{abs}$ & $\Delta H_{rel}$ & \\
        \hline
        10  & 3.96E+13  & 3.58E+13     & -3.79E+12        & -9.57E-02        & 4.52E+01      & 4.50E+01      & -1.45E-01        & -3.21E-03          \\
        20  & 1.69E+37  & 1.61E+37     & -8.26E+35        & -4.88E-02        & 1.24E+02      & 1.24E+02      & -7.20E-02        & -5.84E-04          \\
        30  & 1.98E+65  & 1.91E+65     & -6.49E+63        & -3.28E-02        & 2.17E+02      & 2.17E+02      & -4.80E-02        & -2.22E-04          \\
        40  & 1.85E+96  & 1.81E+96     & -4.58E+94        & -2.47E-02        & 3.20E+02      & 3.20E+02      & -3.60E-02        & -1.13E-04          \\
        50  & 2.57E+129 & 2.51E+129    & -5.08E+127       & -1.98E-02        & 4.30E+02      & 4.30E+02      & -2.90E-02        & -6.71E-05          \\
        60  & 1.91E+164 & 1.88E+164    & -3.16E+162       & -1.65E-02        & 5.46E+02      & 5.46E+02      & -2.40E-02        & -4.41E-05          \\
        70  & 3.96E+200 & 3.90E+200    & -5.61E+198       & -1.42E-02        & 6.66E+02      & 6.66E+02      & -2.10E-02        & -3.09E-05          \\
        80  & 1.41E+238 & 1.39E+238    & -1.75E+236       & -1.24E-02        & 7.91E+02      & 7.91E+02      & -1.80E-02        & -2.28E-05          \\
        90  & 6.07E+276 & 6.00E+276    & -6.70E+274       & -1.11E-02        & 9.20E+02      & 9.19E+02      & -1.60E-02        & -1.74E-05          \\
        100 & 2.39E+316 & 2.36E+316    & -2.38E+314       & -9.95E-03        & 1.05E+03      & 1.05E+03      & -1.40E-02        & -1.37E-05          \\
    \end{tabular}
    \end{table}

\begin{table}
    \caption{\label{tab.Stir}Value of $N!, \log_2 (N!)$ and its corresponding approximation using Stirling's formula. $\Delta$ denotes the error, while the texts of subscripts indicate the relative error or absolute error, respectively.}
    \begin{tabular}{lccccccccc}
        $N$ & $N!$      & $S_t(N)$  & $\Delta S_{t-abs}$ & $\Delta S_{t-rel}$ & $\log (N!)$ & $\log (S_t(N))$ & $\Delta \log_{t-abs}$ & $\Delta \log_{t-rel}$ & \\
        \hline
        10  & 3.63E+06  & 3.60E+06  & -3.01E+04          & -8.30E-03          & 2.18E+01    & 2.18E+01        & -1.20E-02             & -5.52E-04               \\
        20  & 2.43E+18  & 2.42E+18  & -1.01E+16          & -4.20E-03          & 6.10E+01    & 6.11E+01        & -6.01E-03             & -9.84E-05               \\
        30  & 2.65E+32  & 2.65E+32  & -7.36E+29          & -2.80E-03          & 1.08E+02    & 1.08E+02        & -4.01E-03             & -3.72E-05               \\
        40  & 8.16E+47  & 8.14E+47  & -1.70E+45          & -2.10E-03          & 1.59E+02    & 1.59E+02        & -3.01E-03             & -1.89E-05               \\
        50  & 3.04E+64  & 3.04E+64  & -5.06E+61          & -1.70E-03          & 2.14E+02    & 2.14E+02        & -2.40E-03             & -1.12E-05               \\
        60  & 8.32E+81  & 8.31E+81  & -1.15E+79          & -1.40E-03          & 2.72E+02    & 2.72E+02        & -2.00E-03             & -7.36E-06               \\
        70  & 1.20E+100 & 1.20E+100 & -1.43E+97          & -1.20E-03          & 3.32E+02    & 3.32E+02        & -1.72E-03             & -5.17E-06               \\
        80  & 7.16E+118 & 7.15E+118 & -7.45E+115         & -1.00E-03          & 3.95E+02    & 3.95E+02        & -1.50E-03             & -3.81E-06               \\
        90  & 1.49E+138 & 1.48E+138 & -1.37E+135         & -9.00E-04          & 4.59E+02    & 4.59E+02        & -1.34E-03             & -2.91E-06               \\
        100 & 9.33E+157 & 9.33E+157 & -7.77E+154         & -8.00E-04          & 5.25E+02    & 5.25E+02        & -1.20E-03             & -2.29E-06               \\
    \end{tabular}
    \end{table}

As shown in Tab.~\ref{tab.RPSscale} and Tab.~\ref{tab.Stir}, As $N$ increases, $S(N)$ and $S_{lim}(N)$ will grow at a rate close to $(N!)^2$, but both of them maintain the same number of digits. Besides, the absolute error between $S_{lim}(N)$ and $S{N}$ is roughly \textbf{one-hundredth} of $S{N}$, while the absolute error between $N!$ and Stirling's estimation is approximately \textbf{one-thousandth} of $N!$. When taking the logarithmic operation on them, both absolute and relative errors are rapidly reduced to an acceptable range.

Based on the above, it can be concluded that the proposed approximation to maximum RPS entropy is near as good as Stirling's approximation to factorial. When considering the form, the proposed approximation is much more concise compared with Stirling's formula.

\begin{example}\label{ex.SEDERPS}
Suppose a finite set with $N$ elements: $\Omega = \{\chi_1, \chi_2,\ldots, \chi_N  \}$. Its corresponding sample space, power set, and the PES can be marked as $\Omega, 2^\Omega, PES(\Omega)$, respectively.
\end{example}

Then $H_{max-SE}$, $H_{max-DE}$, and $H_{max-RPS}$ can be obtained by the following equations:

\begin{align}
    H_{max-SE}  & =\log(N),       \\
    H_{max-DE}  & =\log(3^N-2^N), \\
    H_{max-RPS} & =\log(S(N)).
\end{align}

In comparison to $H_{max-RPS}$, let's define $H_{lim-RPS}$ as an approximation to maximum RPS entropy and its relative error: $\Delta H_{RPS}$:

\begin{gather}
    H_{lim-RPS}=\log(e \times (N!)^2), \\
    \Delta H_{RPS} = \frac{H_{lim-RPS}-H_{max-RPS}}{H_{max-RPS}} \times 100\%.
\end{gather}

When $N$ increases, the different result of $H_{max-SE}, H_{max-DE}, H_{max-RPS}$ and $H_{lim-RPS}$ are shown in Tab.~\ref{tab.DESERPS} and Fig.~\ref{fig.DESERPS}.
And Fig.~\ref{fig.abs_rel_error} shows the relative error and absolute error between $H_{lim-RPS}$ and $H_{RPS}$.

As shown in Fig.~\ref{fig.DESERPS}, $H_{max-RPS}$ and $H_{lim-RPS}$ exhibit a greater slope, i.e., a higher growth rate, than $H_{max-SE}$ and $H_{max-DE}$, while $H_{max-DE}$ also exhibit a higher increasing speed in comparison to $H_{max-SE}$. This can be clarified through the simple fact that the uncertainty in a finite set's $PES$ is substantially larger than in its power set or sample space for a given number of elements. This is because RPS considers all permutations of the finite set, while DSET and probability theory do not. DSET, on the other hand, considers all potential subsets of the finite set and may thus be considered an extension of probability theory.

By comparing $H_{max-RPS}$ and $H_{lim-RPS}$ in Tab.~\ref{tab.DESERPS}, Fig.~\ref{fig.DESERPS} and Fig.~\ref{fig.abs_rel_error}, it's clear that when $N > 7$, the proposed approximation will converge to $H_{max-RPS}$ quickly. When $N>15$, the accuracy of the approximation will reach two decimal errors. Considering the complex computational steps required by the original calculation process, this approximation provides a more reasonable estimation of the maximum RPS entropy. In contrast, the estimation formula of the Stirling formula for the factorial is not as concise as the proposed estimation formula for the maximum RPS entropy.

\begin{figure}
    \centering
    \includegraphics[width=0.9\textwidth]{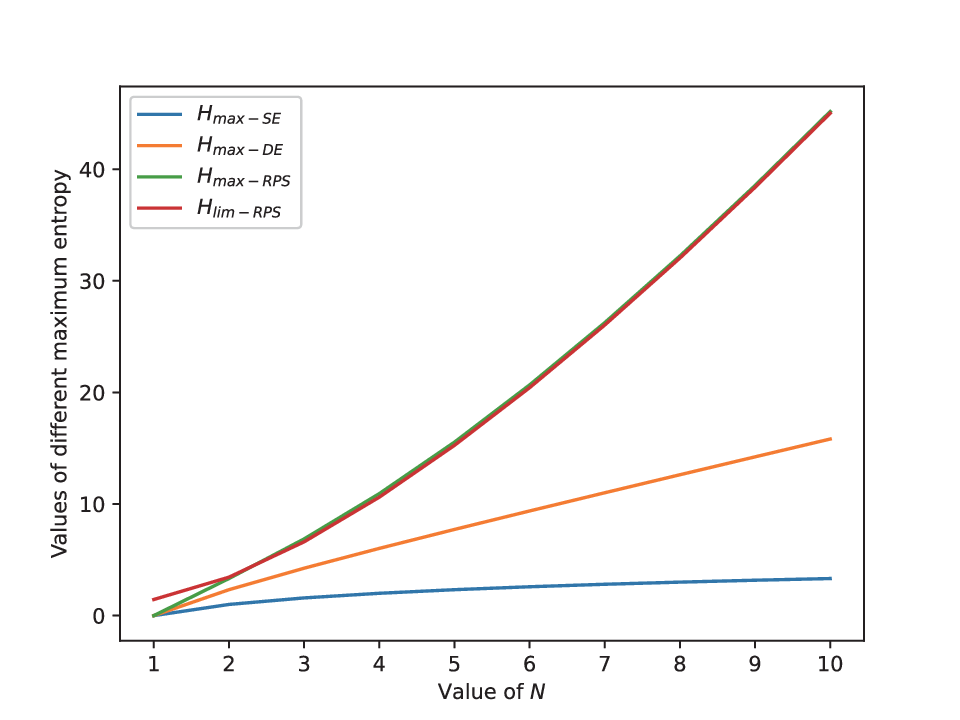}
    \caption{\label{fig.DESERPS}The trend of maximum Shannon entropy, Deng entropy, RPS entropy and the proposed approximation of maximum RPS entropy when $N$ changes.}
\end{figure}

\begin{figure}
    \centering
    \includegraphics[width=0.9\textwidth]{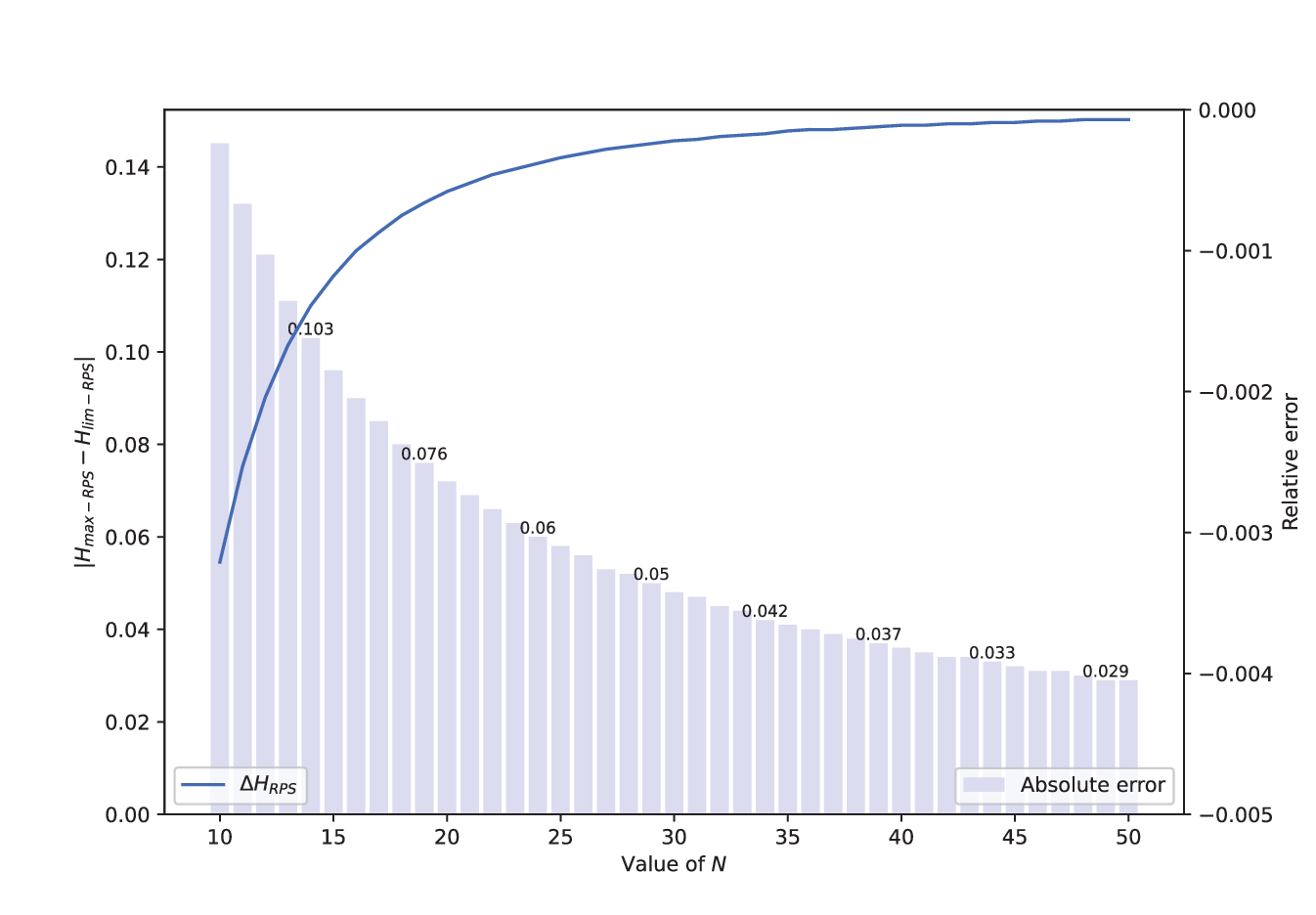}
\caption{\label{fig.abs_rel_error}The trend of relative error and absolute error of the proposed approximation of maximum RPS entropy when $N$ changes. The line chart is denoted as the relative error while the bar chart is denoted as the absolute error.}
\end{figure}

\begin{table}
    \caption{\label{tab.DESERPS}The maximum Shannon entropy, Deng entropy, RPS entropy and the presented result of limit of maximum RPS entropy with different values of $N$.}
    \begin{tabular}{lccccc}
        $N$ & $H_{max-SE}$ & $H_{max-DE}$ & $H_{max-RPS}$ & $H_{lim-RPS}$ & $\Delta H_{RPS}$ \\
    \hline
    1   & 0.00000      & 0.00000      & 0.00000       & 1.44270       & 0.00\%           \\
    2   & 1.00000      & 2.32193      & 3.32193       & 3.44270       & 3.64\%           \\
    3   & 1.58496      & 4.24793      & 6.87036       & 6.61262       & -3.75\%          \\
    4   & 2.00000      & 6.02237      & 10.92780      & 10.61260      & -2.88\%          \\
    5   & 2.32193      & 7.72110      & 15.54060      & 15.25650      & -1.83\%          \\
    6   & 2.58496      & 9.37721      & 20.66910      & 20.42640      & -1.17\%          \\
    7   & 2.80735      & 11.00770     & 26.24950      & 26.04110      & -0.79\%          \\
    8   & 3.00000      & 12.62230     & 32.22310      & 32.04110      & -0.56\%          \\
    9   & 3.16993      & 14.22660     & 38.54240      & 38.38100      & -0.42\%          \\
    10  & 3.32193      & 15.82440     & 45.16990      & 45.02480      & -0.32\%          \\
    \end{tabular}
    \end{table}

\begin{example}
    In this example, the computational complexity of the \textit{envelope} of RPS entropy, as well as the presented limit, is given.
\end{example}

According to \autoref{th.envelope} and Eq.(\ref{eq.SimS_{A}(N)}), the \textit{envelope} of RPS entropy is 

\begin{align}
    C_e (H_{RPS})&=\sum\limits_{a=1}^{N}\left[ A_{N}^{a} (S_{A}(a)-1)\right] \nonumber\\
   &= \sum\limits_{a=1}^{N} \left[ \frac{N!}{(N-a)!}\left( \lfloor e \times a! \rfloor-1 \right) \right].
\end{align}

The function performs a summation from $a=1$ to $N$, resulting in $N$ iterations and complexity of $O(N)$ or linear time. Within the loop, there are several operations that need to be considered:

\begin{itemize}
    \item The computation of $N!/(N-a)!$ using the factorial function, which can be done in $O(N)$ time.
    \item The computation of $\lfloor e \times a!\rfloor$ using the factorial and floor functions, also taking $O(a)$ time.
    \item Subtracting $1$, which is a constant time operation $O(1)$.
    \item Multiplying the results of the above operations, which is also a constant time operation $O(1)$.
\end{itemize}

Therefore, the operations within the loop have an asymptotic time complexity of $O(N)$. Since there are $N$ iterations, the overall asymptotic complexity becomes $O(N) * O(N) = O(N^2)$.

Regarding the presented limit of the \textit{envelope} of RPS entropy $C_{e(lim)}(H_{RPS}=e \times (n!)^2$, its computational complexity is clearly $O(N)$. Considering the accuracy of this approximation, as demonstrated in Example \ref{ex.SEDERPS} and Example \ref{ex.stirling}, the computational efficiency gained from this approximation will be a significant advantage for future applications. For instance, De Gregirui et al. \citep{degregorio2022improved} proposed an estimation of Shannon entropy to quantify the memory of a given system. And Muhammed Rasheed Irshad et al. \citep{irshad2024estimation} considered an estimation of weighted extropy and used it for reliability modeling.

\section{Conclusion}\label{se.colclusion}

RPS is a great extension of DSET. Though the maximum entropy of RPS is presented, its computational complexity makes it difficult to discuss the \textit{envelope} of RPS entropy. This study addressed this issue by presenting the limit form of the \textit{envelope} of RPS entropy. The result $e \times (N!)^2$ establishes a fascinating link between the natural constant $e$ and the factorial function, two fundamental concepts in mathematics.

In future research, there are mainly two problems. The first point of interest is exploring the physical meaning of $e \times (N!)^2$ and its potential correlation with RPS. Additionally, while an approximation of the maximum RPS entropy has been proposed, further research is needed to explore its practical applications in specific domains.

\section*{AUTHOR DECLARATIONS}

\subsection*{Conflict of Interest}
The authors have no conflicts to disclose.

\subsection*{Author Contributions}
\textbf{Jiefeng Zhou}: Conceptualization, Methodology, Formal analysis, Investigation, Writing-original draft, Writing-review \& editing. \textbf{Zhen Li}: Validation. \textbf{Kang Hao Cheong}: Validation.\textbf{Yong Deng}: Writing-review \& editing, Supervision, Project administration, Funding acquisition.

\subsection*{ACKNOWLEDGMENTS}
The work is partially supported by National Natural Science Foundation of China (Grant No. 62373078).

\subsection*{Data Availability}
Data sharing is not applicable to this article as no new data were created or analyzed in this study.


\bibliography{sn-bibliography}

\end{document}